\documentclass[graybox]{svmult}

\usepackage{type1cm}        % activate if the above 3 fonts are
\usepackage{makeidx}         % allows index generation
\usepackage{graphicx}        % standard LaTeX graphics tool
\usepackage{multicol}        % used for the two-column index
\usepackage[bottom]{footmisc}% places footnotes at page bottom

\usepackage{newtxtext}       % 
\usepackage[varvw]{newtxmath}       % selects Times Roman as basic font

\makeindex             % used for the subject index
\begin{document}

\title*{Some Remarks on Multisymplectic and Variational Nature of Monge-Ampère Equations in Dimension Four}
\titlerunning{Some Remarks on Multisymplectic and Variational Nature of M-A Equations in 4D} 
% for an abbreviated version of your contribution title if the original one is too long
\author{Radek Suchánek}
\institute{R. Suchánek  \at Department of Mathematics and Statistics, Masaryk University, Building 08, Kotlářská 2, 611 37 Brno, Czech Republic, 
\at Angevin Laboratory of Mathematical Research - UMR CNRS 6093, University of Angers, 2 Boulevard Lavoisier, 49045, Angers CEDEX 0, France. \\ \email{r.suchanek.r@gmail.com}}
\maketitle

%%%%%%%%%%%%%%%%%%%%%%%%%%%%%%%%%%%%%%%%%%%%
\abstract{
We describe a necessary condition for the local solvability of the strong inverse variational problem in the context of Monge-Ampère partial differential equations and first-order Lagrangians. This condition is based on comparing effective differential forms on the first jet bundle. To illustrate and apply our approach, we study the linear Klein-Gordon equation, first and second heavenly equations of Plebański, Grant equation, and Husain equation, over a real four-dimensional manifold. Two approaches towards multisymplectic formulation of these equations are described.}

%%%%%%%%%%%%%%%%%%%%%%%%%%%%%%%%%%%%%%%%%%%%%%%%%%%%%%%%%%%%%%%%%%%%%%%%%%%%%%%%%%%%%%%%
%%%%%%%%%%%%%%%%%%%%%%%%%%%%%%%%%%%%%%%%%%%%%%%%%%%%%%%%%%%%%%%%%%%%%%%%%%%%%%%%%%%%%%%%

\section*{Introduction}

Since the 19th and early 20th century work of mathematicians such as Joseph Liouville, Gaston Darboux, Sophus Lie, Élie Cartan et al., it is well-known that geometry plays an essential role in the study of ordinary and partial differential equations (PDEs). 

A special subclass of all non-linear second-order PDEs is Monge-Ampère (M-A) equations. They arise in many examples and have numerous applications throughout mathematics and mathematical physics. One can find them in differential geometry of surfaces, hydrodynamics, acoustics, integrability of various geometric structures, variational calculus, Riemannian, CR, and complex geometry, quantum gravity, and even in theoretical meteorology (semi-geostrophic and quasi-geostrophic theory). Many other instances can be listed. For a detailed exposition of interesting applications of M-A equations, particularly in 2D and 3D, see \cite{kushner_lychagin_rubtsov_2006}. 

In this paper, we are mainly interested in the variational structure of M-A equations. In particular, we study whether we can view them as E-L equations for some first-order Lagrangians. Our approach is based on the idea of V. Lychagin to connect the M-A operators with symplectic (on $T^*M$) and contact (on $J^1M)$ geometries. He also defined a class of variational problems related to M-A equations \cite{Lychagin, kushner_lychagin_rubtsov_2006}. 

Afterwards, we observe the relation between M-A equations and multisyplectic geometry, using the results of two slightly different approaches proposed by F. Hélein \cite{multisymplectic-formalims-and-covariant-phase-space}, and D. Harrivel \cite{Dika}. We have found some new aspects which could shed light on this connection. We applied our observations in the context of the following 4D PDEs, very famous for their applications in geometry and theoretical physics related to Einstein gravity and relativistic field theories - Plebański heavenly equations and Klein-Gordon equation. We also considered Grant and Husain equations, which are very close to Plebański second equation.

In 1975, J.F. Plebański introduced his first and second heavenly equations \cite{Plebanski:1975wn}, which belong to the class of M-A equations in 4D. Their close relatives, Grant and Husain equations, were introduced more recently \cite{Grant1993, Husain1993SelfdualGA}. These equations appeared firstly in Einstein gravity, and later were studied by numerous authors, both physicists and mathematicians \cite{Plebanski:1975wn, Grant1993, Husain1993SelfdualGA, Sheftel-Multi-hamiltonian-Plebanski2nd, BANOS2011}. Another significant example of M-A equation is the Klein-Gordon equation, which is a non-homogeneous relativistic wave equation. It was derived in the first quarter of the 20th century by O. Klein and later reformulated in a more compact form by W. Gordon \cite{Gordon}. The underlying structure of this equation can be found in more general situations than scalar fields, and the knowledge of its solutions is relevant in the relativistic perturbative quantum field theory \cite{Klein-Gordon-in-QFT}. The specific form of all the above equations and some further details about them is given below.

%%%%%%%%%%%%%%%%%%%

In the first section, we define M-A operators and related notions, which will be our main tools in working with M-A equations via differential forms. We also recall the contact and symplectic calculus over $J^1M$, which we greatly utilize in our computations. The second section describes the construction of the Euler operator on $\Omega^n (J^1M)$ and its relation to variational problems. In the third section, a~necessary condition for local solvability of the strong inverse variational problem of a~given M-A equation is formulated, together with the corresponding analysis of the aforementioned five M-A equations in four real dimensions. In the fourth section, we present two multisymplectic approaches and provide certain comparison of them, in the context of concrete M-A equations under consideration.

%%%%%%%%%%%%%%%%%%%%%%%%%%%%%%%%%%%%%%%%%%%%%%%%%%%%%%%%%%%%
\newpage 

In the sequel, we will be working with smooth real-valued functions $\phi \in C^\infty(M)$ and their first prolongations $j^1\phi \colon M \to J^1M$, where $J^1M \to J^0M = M \times \mathbb{R} $ is the first jet bundle of $\operatorname{pr}_1 \colon M \times \mathbb{R} \to M $.

A~second-order partial differential equations which are given as a $C^\infty(J^1M)$-linear\footnote{By $C^\infty(J^1M)$-linear we mean that the coefficients can be smooth functions and their first derivatives.} combination of minors of the Hessian matrix $(\phi_{\mu \nu})_{\mu, \nu}$ are called Monge-Ampère equations\footnote{Note that the minors of rank 1 recover all the second-order semi-linear differential equations, whilst the higher order minors (including the determinant of the whole matrix) add specific non-linear terms.} \cite{Lychagin, kushner_lychagin_rubtsov_2006, BANOS2011, Dika}. Consequently, every such equation can be represented by a differential $n$-form on $J^1M$ via M-A operator $\Delta_\omega \phi : = (j^1\phi)^* \omega$. Moreover, one can use \emph{effective} differential forms, which represent M-A equations uniquely (up to a multiple of a non-vanishing function), and without terms corresponding to trivial equations \cite{Lychagin, kushner_lychagin_rubtsov_2006, BANOS2011}. Effective forms on the first jet space, which produce first-order Lagrangians on the base manifold, have a particularly simple local expression. Their image under the Euler operator represents the Euler-Lagrange (E-L) equations \cite{kushner_lychagin_rubtsov_2006}. This feature of the Euler operator, together with the fact that it preserves the effective forms, enables us to study the existence of a first-order Lagrangian for a given M-A equation on the level of differential forms over $J^1M$. Additionally, some effective forms give rise (in a non-unique way) to multisymplectic forms \cite{Dika}. This may happen even for an effective form that comes from a M-A equation which does not have a first-order Lagrangian. Since the multisymplectic reformulation usually starts with a Lagrangian \cite{multisymplectic-formalims-and-covariant-phase-space, ObservableFormsAndFunctionals, firstOrderTheoriesAndPremultisymplectic}, this seems to be an interesting property. We will apply the formalism on the following M-A equations: Plebański heavenly, Grant, Husain, and Klein-Gordon equations. We will consider these equations in the real 4D case. 
%%%%%%%%%%%%%%%%%%%%%%%%%%%%%%%%%%%%

The heavenly equations of Plebański were first derived in \cite{Plebanski:1975wn} in the form
    \begin{eqnarray*}
        \phi_{13} \phi_{24}  - \phi_{14}\phi_{23}  & = 1  \text{     (1st heavenly equation)} \\
        \phi_{11} \phi_{22}  - (\phi_{12})^2 + \phi_{13} +  \phi_{24} & = 0  \text{     (2nd heavenly equation)} 
    \end{eqnarray*}
using self-dual 2-forms over a complex 4D Riemannian space. The duality here is given by the Hodge star operator. The Grant equation and the Husain equation are both based on the Ashtekar-Jacobson-Smolin (AJS) equations, which are Einstein self-dual equations. The AJS equations were derived in \cite{Ashtekar-Jacobson-Smolin} employing the $3+1$ ADS decomposition of spacetime. They characterize 4D complex metrics with self-dual curvature 2-form. Metrics with self-dual curvature form satisfy the vacuum equations of general relativity since they are Ricci flat. In \cite{Grant1993}, the following equation was introduced 
    \begin{equation*}
        \phi_{11} + \phi_{24} \phi_{13} - \phi_{23} \phi_{14} = 0  \text{     (Grant equation)} 
    \end{equation*}
and subsequently rewritten into a system which enabled the author to construct formal solutions. Notably, \emph{the Grant equation is equivalent with the first heavenly equation of Plebański} \cite{Grant1993}. Another reformulation of the AJS equations was provided in \cite{Husain1993SelfdualGA}, in order to identify AJS with a 2D chiral model, and to provide a~Hamiltonian formulation. The resulting equation
   \begin{equation*}
        \phi_{13} \phi_{24} - \phi_{14} \phi_{23} + \phi_{11} + \phi_{22} = 0 \text{     (Husain equation)} 
    \end{equation*} 
enabled V. Husain to show the existence of infinitely many non-local conserved currents. 
%%%%%%%%%%%%%%%%%%%
Another type of an M-A equation is 
    \begin{equation*}
        \phi_{11} - \phi_{22} - \phi_{33} - \phi_{44} + m^2 \phi^2 = 0 \text{     (Klein-Gordon equation)} 
    \end{equation*}
where $m$ is a constant. The Klein-Gordon equations was derived in various ways, for example by W. Gordon \cite{Gordon}. In its real version, it can be interpreted as an equation of motion for a scalar field without charge over a Lorentzian manifold. A key difference between the aforementioned equations is that the Klein-Gordon equation does not arise from self-duality conditions.
%%%%%%%%%%%%%%%%%%%

%%%%%%%%%%%%%%%%%%%%%%%%%%%%%%%%%%%%%%%%%%%%%%%%%%%%%%%%%%%%%%%%%%%%%%%%%%%%%%%%%%%%%%%%
%%%%%%%%%%%%%%%%%%%%%%%%%%%%%%%%%%%%%%%%%%%%%%%%%%%%%%%%%%%%%%%%%%%%%%%%%%%%%%%%%%%%%%%%

\section{Preliminary notions}

\ \ \ \ In this section we fix the notation and introduce basic definitions and statements relevant to our considerations. In particular, we will define the notion of effective forms, Monge-Ampère operators and Monge-Ampère equations. All our considerations are local. We caution the reader about the standard abuse of notation such us denoting a symplectic form by $\Omega$, and by $\Omega(M)$ the exterior algebra of differential forms over $M$. 

%%%%%%%%%%%%%%%%%%%%%%%%%%%%%%%%%%%%%%%%%%%%%%%%%%%%%%%%%%%%%%%%%%%%%%%%%%%%%%%%%%
\hfill 

We denote by $M$ a smooth $n$-dimensional manifold, $(q^1, \ldots, q^n)$ are local coordinates over an open subset $U \subset M$, $TM$ and $T^*M$ are the tangent and cotangent bundle, respectively. Let $J^1M$ be the space of $1$-jets of smooth functions over $M$, which is an affine bundle over $M \times \mathbb{R}$
    \begin{equation*}
        \pi \colon J^1M \to J^0 M = M \times \mathbb{R}
    \end{equation*}
with typical fiber $T^*M$. It is also a fiber bundle over $M$
    \begin{equation*}
        \operatorname{pr}_1 \circ \pi \colon J^1M \to M \ , 
    \end{equation*}
where $\operatorname{pr}_1 \colon M \times \mathbb{R} \to M$. We denote by $(q^1, \ldots, q^n, u, p_1, \ldots, p_n)$ the induced local coordinates on $J^1M$. The first prolongation of $\phi \in C^\infty (M)$\footnote{Each $\phi \in C^\infty(M)$ defines a section $ M \to M \times \mathbb{R}$, $x \mapsto (x, \phi (x))$.} is a section $j^1 \phi \colon M \to J^1M$, given by $x \mapsto (j^1 \phi)(x) \in J^1 M$. Recall that $(j^1 \phi) (x)$ is an equivalence class of functions which are equal up to the first order in derivatives at $x$. In local coordinates, 
    \begin{equation*}
        j^1\phi = (q^\mu, \phi, \phi_\mu ) \ ,
    \end{equation*}
where $\phi_\mu : = \partial_{q^\mu} \phi : = \frac{\partial \phi}{\partial_{q^\mu}} $ is the partial derivative in the direction of the coordinate $q^\mu$. The pullbacks of coordinate functions on $J^1M$ are
    \begin{align*}
        (j^1 \phi)^* q^\mu & = q^\mu & (j^1 \phi)^* u & = \phi &  (j^1 \phi)^* p_\mu & = \phi_\mu \ ,
    \end{align*}

In the local coordinates, we have the identification $J^1U \cong T^*U \times \mathbb{R} $ (which is not canonical). Most relevant for us is that $J^1M$ is naturally equipped with a contact structure \cite{Lychagin, kushner_lychagin_rubtsov_2006}. For more details about jet bundles and structures on them, see \cite{Natural-operations}.

%%%%%%%%%%%%%%%%%%%%%%%%%%%%%%%%%%%%%%%%%%%%%%%%%%%%%%%%%%%%%%%%%%%%%%%%%%%%%%%%%%

\subsection{Contact structure on $J^1M$}

\begin{definition}\label{contact structure - definition}
Let $\omega \in \Omega^1 (M)$ be non-vanishing. Let $\mathcal{D}\subset TM$ be a distribution given by $\mathcal{D}: = \ker \omega$. Then $\omega$ is called a contact form on $M$, if $d \omega|_{\mathcal{D}} \colon \mathcal{D}\to \mathcal{D}^*$ is non-degenerate. Manifold with a distribution described by a contact form is called a~contact manifold and $\D$ is called a contact structure (or contact distribution) on $M$.
\end{definition}

\begin{remark}
Note that the distribution $\mathcal{D}= \ker \omega$ satisfies $\operatorname{codim} \mathcal{D}= 1$. Moreover, the $1$-form describing $\mathcal{D}$ is not unique. Consider a class of $1$-forms, $[\omega]$, given by $\tilde{\omega} \in [\omega]$ if and only if there is a non-vanishing $f \in C^\infty (M)$ s.t. $\tilde{\omega} = f \omega$. Then every representative of the class $[\omega]$ defines the same distribution $\D$. 
\end{remark}

The first jet space comes equipped with the \emph{Cartan distribution}, which infinitesimally describes the condition that a section of $J^1M \to M$ is obtained as a prolongation of a~function $\phi \in C^\infty(M)$. In the induced coordinates, this requirement can be described by the following \emph{contact form}\footnote{We are using the summation convention of summing over the repeated indices.} 
   \begin{equation}\label{contact 1-form in coordinates}
        \mathfrak{c} = \D u - p_\mu \D q^\mu \ .
    \end{equation}
This $1$-form satisfies the definition \ref{contact structure - definition} and we can describe the Cartan distribution as $\mathcal{C} = \ker \mathfrak{c}$. That is, $J^1M$ is a contact manifold\footnote{Cartan distribution exists also on higher jets but the first jets are special due to $\operatorname{codim} \mathcal{C} = 1$.}. By the Darboux theorem, every contact form on $J^1M$ is locally given by \eqref{contact 1-form in coordinates}. The contact form defines the Reeb vector field, $\chi$, by the following conditions 
    \begin{equation}\label{Reeb field conditions}
        \chi \mathbin{\lrcorner} \D \mathfrak{c} = 0 \text{ and }  \mathfrak{c}(\chi) = 1 \ . 
    \end{equation}
In the local coordinates s.t. \eqref{contact 1-form in coordinates} holds, the Reeb field is of the form $\chi = \partial_u$, which immediately follows from \eqref{Reeb field conditions}. Moreover, since $\operatorname{codim} \mathcal{C} = 1$, we get the following splitting of $T J^1 U$
    \begin{equation*}%\label{splitting of the jets space tangent bundle}
        T J^1 U \cong \mathcal{C} \oplus \operatorname{span}(\chi) \cong \ker \mathfrak{c} \oplus \ker \D \mathfrak{c} \ .
    \end{equation*}

%%%%%%%%%%%%%%%%%%%%%%%%%%%%%%%%%%%%%%%%%%%%%%%%%%%%%%%%%%%%%%%%%%%%%%%%%%%%%%%%%%

\subsection{Symplectic calculus on the Cartan distribution}

Contact form on $J^1M$ gives rise to a \emph{symplectic form} on $\mathcal{C}$. 
\begin{definition}\label{symplectic form - definition}
Let $V$ be a vector space,  $\dim V = 2n$. A symplectic form on $V$ is a~$2$-form $\Omega \in \Lambda^2(V^*)$, which is non-degenerate, i.e. $\Omega^{n} : = \Omega \wedge \ldots \wedge \Omega$ is non-vanishing.
\end{definition}

Consider the $2$-form $\Omega : = \D \mathfrak{c}$ on the contact manifold $J^1M$. Then $\Omega$ is obviously closed. In the chosen coordinates, we have
    \begin{equation}\label{symplectic 2-form in coordinates}
        \Omega  = \D q^\mu \wedge \D p_\mu \ . 
    \end{equation}
Note that $\Omega$ is non-degenerate when restricted to $\mathcal{C}$. This means that $\Omega_x$ is a symplectic form on $\mathcal{C}_x$ at every $x \in M$. Using the symplectic form, we can define various useful operators. This leads to considering the space of differential $k$-forms which are degenerate along the Reeb field $\chi$. We will denote this $C^\infty$-module by 
    \begin{equation}\label{forms degenerating along the Reeb field}
        \Omega^k (\mathcal{C}) : = \{ \alpha \in \Omega^k (J^1U) \ | \ \chi \mathbin{\lrcorner} \alpha = 0 \} \ . 
    \end{equation}
Since the interior product $\mathbin{\lrcorner}$ satisfies the graded Leibniz rule with respect to the wedge product, the space 
    \begin{equation*}
        \Omega (\mathcal{C}) : = \underset{k \leq 0}{\bigoplus} \Omega^k (\mathcal{C}) \subset \Omega (J^1M)
    \end{equation*}
has a graded algebra structure. Using suitable projections, $\Omega (\mathcal{C})$ can be turned into a~differential graded algebra.

\textbf{Projection and projected derivative.} Every $\alpha \in \Omega^k (J^1M)$ can be projected on $\Omega^k (\mathcal{C})$ via the projection $p \colon \Omega^k (J^1 M) \to \Omega^k (\mathcal{C})$, acting on arbitrary $k$-form $\alpha$ as 
    \begin{equation}\label{projection operator}
        p(\alpha) = \alpha - \mathfrak{c} \wedge (\chi \mathbin{\lrcorner} \alpha) \ . 
    \end{equation}
Let us show that $p$ has the claimed properties. Firstly, $p^2 = p$, since
    \begin{equation*}%\label{projection  - interm. computation 1}
        p ( p (\alpha) ) = \alpha - \mathfrak{c} \wedge (\chi \mathbin{\lrcorner} \alpha) - \mathfrak{c} \wedge \bigl( \chi \mathbin{\lrcorner} ( \alpha - \mathfrak{c} \wedge (\chi \mathbin{\lrcorner} \alpha) \bigr)  =  p (\alpha) \ . 
    \end{equation*}
Secondly, $p (\alpha) \in \Omega^k (\mathcal{C}) $, since 
   \begin{equation*}%\label{projection  - interm. computation 2}
        \chi \mathbin{\lrcorner} p(\alpha ) =  \chi \mathbin{\lrcorner} \alpha - \chi \mathbin{\lrcorner} \alpha + (\chi \wedge \chi) \mathbin{\lrcorner} \alpha \wedge \mathfrak{c} = 0 \ .
    \end{equation*}
Note that the property $\alpha \in \Omega (\mathcal{C})$ is not preserved by the exterior derivative $\D \colon \Omega^k (J^1 M) \to \Omega^{k+1} (J^1 M)$. So with the projection $p$, we define the degree $1$ derivation $\D_p$ as the composition
    \begin{equation}\label{projected exterior derivative - definition}
       \D_p: = p  \circ \D \colon \Omega^k (J^1 M) \to  \Omega^{k+1} (\mathcal{C}) \ ,
    \end{equation}

\textbf{Bottom operator.} Since $\Omega$ is non-degenerate on $\mathcal{C}$, the assignment $\xi \mapsto \xi \mathbin{\lrcorner} \Omega$ defines an isomorphism $\iota \colon \mathcal{C} \to \mathcal{C}^*$, which further induces an isomorphism $\Lambda^2 \iota^{-1} \colon \Lambda^2 \mathcal{C}^* \to \Lambda^2 \mathcal{C}$. This enables us to define $X_{\Omega} : = \Lambda^2 \iota^{-1} (\Omega)$. In coordinates,
    \begin{equation*}%\label{2-vector corresponding to symplectic structure}
        X_{\Omega} =  \partial_{q^\mu} \wedge \partial_{p_\mu} \ . 
    \end{equation*}
Contracting with the $2$-vector field $X_\Omega$ leads to the  \emph{bottom operator} $\bot \colon \Omega^k (J^1 U) \to \Omega^{k-2} (J^1 U) $. More precisely, for $k$-form $\alpha$, $k > 1$,
    \begin{equation}\label{bottom operator}
       \bot \alpha : = X_{\Omega} \mathbin{\lrcorner} \alpha \ . 
    \end{equation}
For $k \leq 1$ define $\bot \alpha = 0$. Our convention is such that $\bot \Omega = \partial_{p_\mu} \mathbin{\lrcorner} \partial_{q^\mu} \mathbin{\lrcorner} (\D q^\mu \wedge \D p_\mu) = n $. The motivation for defining the bottom operator will be more apparent in the next paragraphs.

%%%%%%%%%%%%%%%%%%%%%%%%%%%%%%%%%%%%%%%%%%%%%%%%%%%%%%%%%%%%%%%%%%%%%%%%%%%%%%%%%%

\subsection{Monge-Ampère operators and effective forms}

\begin{definition}
Let $\omega \in \Omega^n(J^1M)$ be an arbitrary $n$-form, $n = \dim M$. The Monge-Ampère operator corresponding to $\omega$, $\Delta_\omega \colon C^\infty (M) \to \Omega^n(M)$,
is defined as 
    \begin{equation}\label{Monge-Ampere operator}
        \Delta_{\omega} \phi : = (j^1\phi)^* \omega \ .
    \end{equation}
The differential equation 
    \begin{equation}
        \Delta_\omega \phi = 0
    \end{equation}
is called a Monge-Ampère equation. 
\end{definition}

Notice that the expression $\Delta_\omega \phi = 0$ defines an equation on $M$ only when $\omega$ is a~$n = \dim M$-form. In this way, the M-A operators enable us to represent M-A equations by differential forms. Note that we have a certain ambiguity in this representation due to 
    \begin{equation*}
        (j^1\phi)^* \mathfrak{c} = \D \phi - \phi_\mu \D q^\mu = 0 \ . 
    \end{equation*}
In full generality, this ambiguity is described by an ideal of the exterior algebra over $J^1M$, generated by the contact form and its exterior derivative
    \begin{equation}\label{ideal generated by the contact form}
        \mathcal{I} = < \mathfrak{c} , \D \mathfrak{c} > \subset \Omega (J^1M) \ .
    \end{equation}
Recall that $\Omega (J^1M)$ is a graded algebra, which implies that $\mathcal{I}$ is a graded ideal
    \begin{equation*}
        \mathcal{I}^k : = \mathcal{I} \cap \Omega^k (J^1M) \ .
    \end{equation*}
Thus, the redundancy in M-A equations is given by
    \begin{equation}\label{graded part of ideal generated by the contact form}
        \omega \in \mathcal{I}^n \iff \Delta_{\omega} \phi = 0 \ \forall \phi \ .
    \end{equation}
This suggest to work with the equivalence classes of $\Omega^n(J^1M) / \mathcal{I}^n$ instead of using arbitrary forms in $\Omega^n(J^1M)$ to describe M-A equations on $M$. Nevertheless, such an approach is not very convenient for computations in local coordinates. To avoid this problem, we use the following definition of \emph{effective forms}, which captures the above idea of working with forms which do not contain the redundant terms.
    
\begin{definition}\label{effective forms - definiton via bottom operator}
Let $\omega \in \Omega^k (J^1M)$, $k \leq n$. Then $\omega$ is called effective, if 
    \begin{equation}\label{effective forms - working definition}
        \chi \mathbin{\lrcorner} \omega = 0 \text{ and } \bot \omega = 0 \ .
    \end{equation}
\end{definition}
For further details about effective forms and how the above definition can be linked with the equivalence classes of $\Omega^n(J^1M) / \mathcal{I}^n$, see \cite{Lychagin, kushner_lychagin_rubtsov_2006}.

Recall that $\chi \mathbin{\lrcorner} \omega = 0$ means $\omega \in \Omega^k (\mathcal{C})$ (see \eqref{forms degenerating along the Reeb field}). The conditions \eqref{effective forms - working definition} will be our working definition when dealing with effective forms. Note also that the condition $\bot \omega = 0$ is equivalent to $\Omega \wedge \omega = 0$ if and only if $n = k$. 

\begin{example}{Example 1.3} \label{effective forms - example}
Let $\beta = \D q^1 \wedge \D q^2 \wedge \ldots \wedge \D q^n$ and $ \beta_\mu : = \partial_{q^\mu} \mathbin{\lrcorner} \beta$. Then 
    \begin{equation*}
        \omega = b_\mu \beta_\mu \wedge \D p_\mu + b \beta
    \end{equation*}
is effective for arbitrary choice of $b, b_\mu \in C^\infty(J^1M)$, $\mu = 1, \ldots , n$. Indeed, $\omega$ does not contain the $\D u$ term, hence we have $\chi \mathbin{\lrcorner} \omega = 0$. Next, we have
    \begin{equation*}
        \bot \omega =  b_\mu \bot (\beta_\mu \wedge \D p_\mu ) + b \bot \beta 
    \end{equation*}
due to $C^\infty(J^1M)$-linearity of the interior product $\chi \mathbin{\lrcorner}$. Recall that we use the summation convention, so $\beta_\mu \wedge \D p_\mu$ consists of $n$ terms. The first one is $\beta_1 \wedge \D p_1 =  \D q^2 \wedge \ldots \wedge \D q^n \wedge \D p_1$. The bottom operator gives 
    \begin{equation*}
        \bot (\beta_1 \wedge \D p_1) = (\partial_{q^\mu} \wedge \partial_{p_\mu}) \mathbin{\lrcorner} (\beta_1 \wedge \D p_1) =  \partial_{p_\mu} \mathbin{\lrcorner} \partial_{q^\mu}  \mathbin{\lrcorner} (\beta_1 \wedge \D p_1) = \partial_{q^1} \mathbin{\lrcorner}  \beta_1 = 0 \ . 
    \end{equation*}
Similarly for all the other terms of $\beta_\mu \wedge \D p_\mu$. Obviously $\bot \beta = 0$ since $\beta$ does not contain any $\D p$ term. We see that $\omega$ is effective. Notice that the coefficients of $\omega$ might depend on $u$.  
\end{example}

Important result in the theory of effective forms is the Hodge-Lepage decomposition, proved by V. Lychagin in \cite{Lychagin} using the representation theory of $\mathfrak{sl}_2(\mathbb{R})$. 

\begin{theorem}
Every $\omega \in \Omega^k (\mathcal{C}), k \leq n$, can be written in the form
    \begin{equation}\label{Hodge-Lepage decomposition}
        \omega = \omega_{\epsilon} + x \wedge \Omega \ ,
    \end{equation}
for some $x \in \Omega^{k-2} (\mathcal{C})$ and a uniquely given $\omega_\epsilon \in \Omega^k (\mathcal{C})$ satisfying $\bot \omega_\epsilon = 0$.
\end{theorem}

\begin{corollary}\label{lemma for the necessary condition theorem}
Suppose that $\omega_1, \omega_2 \in \Omega^n (\mathcal{C})$ determine the same Monge-Ampère equation. Then the effective parts satisfy
    \begin{equation}\label{effective forms of the same equation are scalar multiple of each other}
        \omega_{1\epsilon} = k\omega_{2\epsilon}
    \end{equation}
for a non-vanishing function $k \in C^\infty (J^1M)$. 
\end{corollary}

\begin{proof}
Two forms determine the same equation if and only if for all $\phi$
    \begin{equation}\label{intermediate step 3}
         \Delta_{\omega_1} \phi = \tilde{k} \Delta_{\omega_2} \phi \ ,
    \end{equation}
for some non-vanishing $\tilde{k} \in C^\infty(M)$. Notice that $\Delta$ is $C^\infty(J^1 M)$-equivariant in the $\omega$ argument, i.e. for arbitrary $\omega$ and $k \in C^\infty(J^1M)$ we have\footnote{Note that $(j^1\phi)^*k = k \circ j^1\phi$ since $k$ is a function.}
    \begin{equation*}
        \Delta_{k \omega} \phi = \bigl((j^1\phi)^*k\bigr) \Delta_{\omega} \phi  \ .
    \end{equation*}
Moreover, $\Delta$ is $\mathbb{R}$-linear in the lower argument, so for arbitrary $\omega_1, \omega_2$, and all $\phi$
    \begin{equation*}
        \Delta_{\omega_1} \phi - \Delta_{\omega_2} \phi = \Delta_{\omega_1 - \omega_2} \phi \ .
    \end{equation*}
Hence \eqref{intermediate step 3} can be rewritten as
    \begin{equation*}
        \Delta_{\omega_1} \phi - \tilde{k}  \Delta_{\omega_2} \phi =  \Delta_{\omega_1 - k \omega_2} \phi = 0 \ , 
    \end{equation*}
for appropriate $k \in C^\infty(J^1 M)$ s.t. $(j^1\phi)^*k = \tilde{k}$. The above equation holds for all $\phi$ if and only if
    \begin{equation*}
        \alpha : = \omega_1 - k \omega_2 \in \mathcal{I}^n
    \end{equation*}
(see \eqref{graded part of ideal generated by the contact form}). Since every $\alpha \in \mathcal{I}^n$ satisfies $\alpha_\epsilon = 0$ and every $\omega \in \Omega (\mathcal{C})$ satisfies $(k\omega)_\epsilon = k \omega_\epsilon$, we conclude $\omega_{1\epsilon} = k\omega_{2\epsilon}$. 
\end{proof}

Using the projection operator \eqref{projection operator} together with the Hodge-Lepage decomposition, we know that every $k$-form $\omega$ on $J^1M$ has a unique effective part $\omega_\epsilon$ (of the same degree). This means that every M-A equation $\Delta_\omega \phi = 0$ can be represented by a~unique differential form which does not contain terms generating trivial equation. We will use this observation in order to study the variational nature of the PDEs under consideration.

%%%%%%%%%%%%%%%%%%%%%%%%%%%%%%%%%%%%%%%%%%%%%%%%%%%%%%%%%%%%%%%%%%%%%%%%%%%%%%%%%%%%%%%%
%%%%%%%%%%%%%%%%%%%%%%%%%%%%%%%%%%%%%%%%%%%%%%%%%%%%%%%%%%%%%%%%%%%%%%%%%%%%%%%%%%%%%%%%

\section{Lagrangians, variational problems and the Euler operator}

Taking the pullback of a $n$-form on the jet space results in a $n$-form on the base manifold $M$, which can be integrated over $M$. Let $\phi$ be compactly supported, $\omega \in \Omega^n(J^1 M)$. Define the (action) functional corresponding to $\Delta_{\omega} \phi$ by
    \begin{equation}\label{M-A functional}
        \Phi_\omega [\phi] = \int_M \Delta_{\omega} \phi \ .
    \end{equation}

\begin{definition}\label{first-order Lagrangian - definition}
We call an element $\omega \in \Omega^n(J^1 M)$ a Lagrangian. A first-order Lagrangian is a $n$-form $\omega$ such that $\Delta_\omega \phi$ depends on $\phi$ up to the first order. 
\end{definition}
 
%If the context allows for ambiguity, we will refer to $\omega$ as a Lagrangian on the jet space, and the corresponding $j^1\phi^* (\omega)$ will be called Lagrangian on the base maifold. 

%%%%%%%%%%%%%%%%%%%%%%%%%%%%%%%%%%%%%%%
\subsection{First-order Lagrangians.}

We are focused on the first-order Lagrangians as defined in \ref{first-order Lagrangian - definition} because they yield all possible first-order Lagrangian functions on $M$\footnote{after the pullback by $(j^1\phi)^*$ and choice of the volume form on $M$}. The following lemma describes the most general form the first-order Lagrangians can have.

\begin{proposition}\label{first-order Lagrangians - proposition}
Every effective first-order Lagrangian for one scalar field $\phi$ is locally of the form 
    \begin{equation}\label{first-order Lagrangian - general form}
        L \beta = L(q^\mu, u, p_\mu) \D q^1 \wedge \ldots \wedge \D q^n \ .
    \end{equation}
for some $L \in C^\infty (J^1M)$.  
\end{proposition}

\begin{proof}
Let $\omega \in \Omega^n(J^1 M)$ be arbitrary. If $\Delta_\omega \phi$ is assumed to depend on the first derivatives of $\phi$ at most, then $\omega$ cannot contain any $\D p_i$ term. Thus 
    \begin{equation*}
        \omega = L \beta +  L_I \D q^I\wedge \D u \ ,
    \end{equation*}
where $\beta = \D q^1 \wedge \ldots \wedge \D q^n$ and $L, L_I \in C^\infty (J^1M)$ with $I = i_1 \ldots i_{k-1}$ running through all possible combinations s.t. $ 1 \leq i_1 \leq \ldots \leq i_{k-1} \leq n$. Now recall that $\omega$ can still contain some terms resulting in zero after the pullback. Due to the Hodge-Lepage decomposition \eqref{Hodge-Lepage decomposition}, every $\omega$ has a unique effective part $\omega_\epsilon$ and the corresponding functionals satisfy
    \begin{equation*}
        \int_M \Delta_{\omega} \phi = \int_M \Delta_{\omega_\epsilon} \phi \ .
    \end{equation*}
So without loss of generality, we may assume that $\omega$ is effective. This implies two things: $\chi \mathbin{\lrcorner} \omega = 0$ and $\bot \omega = 0$. The first condition rules out the terms containing $\D u$ and we are left with $ \omega = L \beta$. It is easy to check that $\bot L \beta = 0$, meaning that $L \beta$ is effective. Thus we conclude that \eqref{first-order Lagrangian - general form} \emph{is the most general first-order Lagrangian for one scalar field $\phi$, which does not contain any terms that would vanish after the pullback on $M$.} 
\end{proof}

%%%%%%%%%%%%%%%%%%%%%%%%%%%%%%%%%%%%%%%
\subsection{Euler-Lagrange equations and the Euler operator.}

Every functional $\Phi_\omega [\phi]$ defines a variational problem $\delta \Phi_\omega [\phi] = 0$ and the corresponding E-L equation. Once we fix a functional, we may compute the E-L equation explicitly. A~natural question at this point is whether we can find $\tilde{\omega} \in \Omega^n (J^1M)$ so that the E-L equation $\delta \Phi_\omega [\phi] = 0$ is given by the Monge-Ampère equation $\Delta_{\tilde{\omega}} \phi = 0$. The answer is positive and $\tilde{\omega}$ can be determined using the \emph{Euler operator} $\mathcal{E}$. 

\begin{definition}
Euler operator $\mathcal{E} \colon \Omega^n (J^1 M) \to \Omega^n (J^1 M)$, $n = \dim M$ is defined by 
    \begin{equation}\label{Euler operator}
        \mathcal{E} : = \D_p  \bot  \D_p + \mathcal{L}_{\chi} \ , 
    \end{equation}
where $\D_p$ is defined by \eqref{projected exterior derivative - definition}, $\bot$ is defined by \eqref{bottom operator}, and $\mathcal{L}_{\chi}$ is the Lie derivative along the Reeb field given by \eqref{Reeb field conditions}.
\end{definition}

The key motivation for us to work with the Euler operator is the following equivalence
    \begin{equation}\label{Euler operator and vartiational problems}
        \delta  \Phi_\omega [\phi] = 0 \iff \Delta_{\mathcal{E}(\omega)} \phi = 0 \ .
    \end{equation}
In other words, the variational problem given by functional of $\omega$ is described by $\mathcal{E}(\omega)$. The proof of this statement and many other useful properties, as well as the details about the cohomological origin of the defining equation \eqref{Euler operator} can be found in \cite{Lychagin, kushner_lychagin_rubtsov_2006}. 

We have the following lemma, which will be used to formulate the necessary conditions for the existence of a first-order Lagrangian of a given PDE (i.e. necessary conditions for the existence of a solution to a given local inverse variational problem).

\begin{lemma}\label{lemma about the Euler operator end E-L equations}
Let $L \beta \in \Omega^n(J^1M)$ be a first-order Lagrangian, $\mathcal{E}$ be defined by \eqref{Euler operator}. Then 
    \begin{enumerate}
        \item $\mathcal{E} (L\beta)$ is effective. 
        \item $\Delta_{\mathcal{E}(L\beta)} \phi = 0$ is the E-L equation of $\Phi_{L\beta} [\phi]$. 
    \end{enumerate}
\end{lemma}

\begin{proof}
Assume the local coordinates satisfying \eqref{contact 1-form in coordinates} and observe that $\chi \mathbin{\lrcorner} L\beta = L \chi \mathbin{\lrcorner} \beta = 0$. Direct computation gives
    \begin{equation*}
         \D_p \bot \D_p (L\beta) =  \frac{\partial^2 L}{\partial p_{\mu} \partial p_{\nu}} \beta_{\mu} \wedge \D p_{\nu}  - ( \frac{\partial^2 L}{\partial q^{\mu} \partial p_{\mu}} +  p_{\mu} \frac{\partial^2 L}{\partial u \partial p_{\mu}} ) \beta \ ,
    \end{equation*}
where $\beta_\mu : = \partial_{q^\mu} \mathbin{\lrcorner} \beta = \partial_{q^\mu} \mathbin{\lrcorner} (\D q^1 \wedge \ldots \wedge \D q^n)$. Using the Cartan formula $\mathcal{L} = \mathbin{\lrcorner} d + d \mathbin{\lrcorner}$, we further obtain
    \begin{equation*}
        \mathcal{L}_{\chi} (L\beta) =  \frac{\partial L}{\partial u} \beta + L( \chi \mathbin{\lrcorner} \D \beta + \D \chi \mathbin{\lrcorner} \beta) =  \frac{\partial L}{\partial u} \beta \ .
    \end{equation*}
Thus, following the definition \eqref{Euler operator}, the coordinate expression of $\mathcal{E}(L\beta)$ is
    \begin{equation}\label{E (L beta) in local coordinates}
       \mathcal{E} (L \beta) = \frac{\partial^2 L}{\partial p_\mu \partial p_\nu} \beta_\mu \wedge \D p_\nu - ( \frac{\partial^2 L}{\partial q^\mu \partial p_\mu } +  p_{\mu} \frac{\partial^2 L}{\partial u \partial p_\mu} - \frac{\partial L}{\partial u}) \beta \ .
    \end{equation}
Let us denote $ B_{\mu \nu }: = \frac{\partial^2 L}{\partial p_\mu \partial p_\nu}$ and $\beta_{\mu \nu}: = (\partial_{q^\mu} \wedge \partial_{q^\nu}) \mathbin{\lrcorner} \beta $. Hence $B_{\nu \mu} = B_{\mu \nu}$, and $\beta_{\nu \mu} = - \beta_{\mu \nu}$. We will check that $\mathcal{E} (L \beta)$ is effective (see def. \ref{effective forms - working definition}). Firstly recall that $\chi = \partial_u$ and that \eqref{E (L beta) in local coordinates} does not contain $\D u$, so $\chi \mathbin{\lrcorner} \mathcal{E} (L \beta) = 0$. Secondly, since $\bot \beta = 0$,     
    \begin{equation*}
         \bot \mathcal{E}(L\beta) = B_{\mu \nu} ( \partial_{q^\alpha} \wedge \partial_{p_\alpha} ) \mathbin{\lrcorner} (\D p_{\nu} \wedge \beta_{\mu}) = - B_{\mu \nu} \partial_{q^\nu} \mathbin{\lrcorner} \beta_{\mu}  = 
            \begin{cases} 
              - B_{\mu \nu} \beta_{\mu \nu} & \mu = \nu \\
              0 & \mu \neq \nu 
            \end{cases} \ . 
    \end{equation*}
Writing the sums over $\mu, \nu$ explicitly, the term $ B_{\mu \nu} \beta_{\mu \nu}$ reads as
    \begin{equation*}
        B_{\mu \nu} \beta_{\mu \nu} = \sum_{\mu < \nu} (B_{\mu \nu} \beta_{\mu \nu} + B_{\nu \mu} \beta_{\nu \mu} ) =  \sum_{\mu < \nu} B_{\mu \nu} (\beta_{\mu \nu} - \beta_{\mu \nu}) = 0 \ ,  
    \end{equation*}
which implies $\bot \mathcal{E}(L\beta) = 0$.  

To show the latter statement, we firstly notice that $\mathcal{E}$ is a $0$ degree operator, which follows directly from $\deg \D_p = 1, \deg \bot = -2, \deg \mathcal{L} = 0$. Hence starting with $L\beta \in \Omega^n(J^1M)$, the result $\mathcal{E}(L\beta)$ is also a $n$-form and $\Delta_{\mathcal{E}(\omega)} \phi = 0$ is a well-defined equation on $M$. The property \eqref{Euler operator and vartiational problems} is then expressed for $\omega = L\beta$ as follows
    \begin{equation*}
       \delta \Phi_{L\beta} [\phi] = \delta \int_M (j^1\phi)^* L \beta = 0 \iff \Delta_{\mathcal{E}(L\beta)} \phi = 0 \ .
    \end{equation*}
Using the coordinate description of $\mathcal{E}(L\beta)$ given by \eqref{E (L beta) in local coordinates}, we get 
    \begin{equation*}
        \Delta_{\mathcal{E}(L\beta)} \phi = 0 \iff \frac{\partial (j^1\phi)^*L}{\partial \phi} - \frac{\partial}{\partial q^\mu } \frac{\partial (j^1\phi)^*L}{\partial \phi_\mu } = 0 \ ,
    \end{equation*}
which is the standard form of the E-L equation for a first-order Lagrangian function $(j^1\phi)^*L = L(q^\mu, \phi, \phi_\mu)$ on $M$, corresponding to $\Phi_{L\beta} [\phi] = \int_M \Delta_{L\beta} \phi$. 
\end{proof}

%%%%%%%%%%%%%%%%%%%%%%%%%%%%%%%%%%%%%%%%%%%%%%%%%%%%%%%%%%%%%%%%%%%%%%%%%%%%%%%%%%%%%%%%
%%%%%%%%%%%%%%%%%%%%%%%%%%%%%%%%%%%%%%%%%%%%%%%%%%%%%%%%%%%%%%%%%%%%%%%%%%%%%%%%%%%%%%%%

%%%%%%%%%%%%%%%%%%%%%%%%%%%%%%%%%%%%%%%%%
\section{Effective forms and the inverse variational problem}
%%%%%%%%%%%%%%%%%%%%%%%%%%%%%%%%%%%%%%%%%

In this section, we will see how M-A equations can be described by effective forms, which provide a unique (up to a scalar multiple) representation of the equation by a~differential form on the first jet space\footnote{The equation can be reconstructed from the differential form via the M-A operator \eqref{Monge-Ampere operator}.}. This enables us to show that both Plebański heavenly, Husain and Grant equations do not have a first-order Lagrangian which would solve the corresponding (local) inverse variational problem. 

The first and easy step is to find a simple representation of the equation (see def. \ref{simple representation of M-A equation}). The simple representation might not be effective. Indeed, this is the case in all the aforementioned equations. The Hodge-Lepage decomposition \eqref{Hodge-Lepage decomposition} assures that we can always find the effective part of a given form, although it does not give a recipe for doing so. Thus we introduce lemma \ref{lemma - computation of effective part} which provides an efficient algorithmic way to determine the effective form of a M-A equation in the case $\dim M = 4$. The following lemma is an intermediate step. 

\begin{lemma}
Let $\omega \in \Omega^2 (\mathcal{C})$ be arbitrary and $\Omega = d \mathfrak{c}$ be the symplectic form on the contact structure $\mathcal{C} \subset T (J^1M)$. The following holds
    \begin{equation}\label{perp applied to a wedge}
        \bot (\omega \wedge \Omega ) = (\bot \omega) \Omega + (n-2) \omega \ ,
    \end{equation}
where $n = \dim M$.   
\end{lemma}

\begin{proof}
Recall that, in the local coordinates s.t. \eqref{contact 1-form in coordinates} holds, we have $\Omega = \D q^\mu \wedge \D p_\mu$ and $\bot \omega = (\partial_{q^\mu} \wedge \partial_{p_\mu}) \mathbin{\lrcorner} \omega = \partial_{p_\mu} \mathbin{\lrcorner} \partial_{q^\mu} \mathbin{\lrcorner} \omega$, which implies $\bot \Omega = n$. Hence
    \begin{equation}\label{intermediate computation 3.1}
        \bot (\omega \wedge \Omega ) = (\bot \omega) \Omega - \partial_{q^\mu} \mathbin{\lrcorner} \omega \wedge \partial_{p_\mu} \mathbin{\lrcorner} \Omega + \partial_{p_\mu} \mathbin{\lrcorner} \omega \wedge \partial_{q^\mu} \mathbin{\lrcorner} \Omega + n \omega \ . 
    \end{equation}
We will show that the middle two terms add up to $-2\omega$. Note that the basis of $\Omega^2 (\mathcal{C})$ consists of pairs $\D q^\mu \wedge \D q^\nu, \D q^\mu \wedge \D p_\nu, \D p_\mu \wedge \D p_\nu $. Because $\partial_q, \partial_p$ are duals to $\D q, \D p$, the basis of $\Omega^2 (\mathcal{C})$ satisfies
    \begin{align*}
        \partial_{q^\mu} \mathbin{\lrcorner} (\D q^\nu \wedge \D q^\xi) & = \delta_{\mu \nu} \D q^\xi - \delta_{\mu \xi} \D q^\nu \ , &  \partial_{p_\mu} \mathbin{\lrcorner} (\D q^\nu \wedge \D q^\xi) & = 0  \ , \\ 
        \partial_{q^\mu} \mathbin{\lrcorner} (\D q^\nu \wedge \D p_\xi) & = \delta_{\mu \nu} \D p_\xi \ , & \partial_{p_\mu} \mathbin{\lrcorner} (\D q^\nu \wedge \D p_\xi) & = - \delta_{\mu \xi} \D q^\nu \ ,  \\
        \partial_{q^\mu} \mathbin{\lrcorner} (\D p_\nu \wedge \D p_\xi) & = 0 \ , & \partial_{p_\mu} \mathbin{\lrcorner} (\D p_\nu \wedge \D p_\xi) & = \delta_{\mu \nu} \D p_\xi - \delta_{\mu \xi} \D p_\nu \ .
    \end{align*}
Since every $\omega \in \Omega^2 (\mathcal{C})$ is of the form $ \omega = \omega_{IJ} \D q^I \wedge \D p_J$ for some functions $ \omega_{IJ} \in C^\infty (J^1M)$, where $I,J$ are ascending multiindices of appropriate length. Due to $C^\infty$-linearity of $\mathbin{\lrcorner}$, we can, without loss of generality, assume that all $\omega_{IJ}$ are constant functions, say $\omega_{IJ} = 1$, and write
    \begin{equation*}
        \omega = \sum_{\nu < \xi } \D q^\nu \wedge \D q^\xi + \sum_{\nu, \xi} \D q^\nu \wedge \D p_\xi + \sum_{\nu < \xi} \D p_\nu \wedge \D p_\xi \ . 
    \end{equation*}
Using the above relations we obtain 
    \begin{equation*}
        \partial_{q^\mu} \mathbin{\lrcorner} \omega \wedge \partial_{p_\mu} \mathbin{\lrcorner} \Omega =  (\delta_{\mu \nu} \D q^\xi - \delta_{\mu \xi} \D q^\nu + \delta_{\mu \nu} \D p_\xi ) \wedge (- \D q^\mu) = 2 \D q^\nu \wedge \D q^\xi + \D q^\nu \wedge \D p_\xi \ , 
    \end{equation*}
and similarly 
    \begin{equation*}
        \partial_{p_\mu} \mathbin{\lrcorner} \omega \wedge \partial_{q^\mu} \mathbin{\lrcorner} \Omega = (- \delta_{\mu \xi} \D q^\nu + \delta_{\mu \nu} \D p_\xi - \delta_{\mu \xi} \D p_\nu ) \D p_\mu = - \D q^\nu \wedge \D p_\xi - 2 \D p_\nu \wedge \D p_\xi \ .
    \end{equation*}
Combining the last two results to fit the terms in \eqref{intermediate computation 3.1} yields
    \begin{equation*}
        - \partial_{q^\mu} \mathbin{\lrcorner} \omega \wedge \partial_{p_\mu} \mathbin{\lrcorner} \Omega + \partial_{p_\mu} \mathbin{\lrcorner} \omega \wedge \partial_{q^\mu} \mathbin{\lrcorner} \Omega = - 2 \omega \ ,
    \end{equation*}
which proves the formula \eqref{perp applied to a wedge}.
\end{proof}

We use the previous lemma to prove the following. A general formula and its proof can be found in \cite{Lychagin}.  
\begin{lemma}\label{lemma - computation of effective part}
Let $\omega \in \Omega^4 (\mathcal{C})$ be arbitrary, $ n = \dim M > 2$. The effective part $\omega_{\epsilon}$ is given by
    \begin{equation}\label{effective part of a 4-form on the Cartan distribution}
         \omega_{\epsilon} = \omega - \frac{1}{n-2} \bot \omega \wedge \Omega + \frac{\bot^2 \omega}{2(n-1)(n-2)} \Omega \wedge \Omega \ . 
    \end{equation}
\end{lemma}

\begin{proof}
Consider the Hodge-Lepage decomposition
    \begin{equation*}%\label{eq1: decomposition of omega}
        \omega = \omega_\epsilon + x \wedge \Omega \ ,
    \end{equation*}
where $\omega_\epsilon \in \Omega^k(\mathcal{C})$ is the unique effective part of $\omega$ and $x \in \Omega^{k-2}(\mathcal{C})$ is not necessarily effective. Applying $\bot$ twice on the above equation together with the formula \eqref{perp applied to a wedge} gives the following system
    \begin{eqnarray*}%\label{eq2: perp (decomposition of omega)}
        \bot \omega & = (\bot x) \Omega + (n-2) x \ , \\ 
        \bot^2 \omega & = 2(n-1) \bot x \ ,
    \end{eqnarray*}
which can be solved for $x$ 
    \begin{equation*}
        x = \frac{1}{(n-2)}\bot \omega - \frac{\bot^2 \omega}{2(n-1)(n-2)} \Omega \ .
    \end{equation*}
Substituting this into the Hodge-Lepage decomposition yields the formula for the effective part of a $4$-form $\omega$.
\end{proof}

%%%%%%%%%%%%%%%%%%%%%%%%%%%%%%%%%%%%%%%%%%%%%

A differential $k$-form is called \emph{simple} if it contains only one summand, when expressed in the canonical coordiantes \eqref{contact 1-form in coordinates}. For example, let $k=2$. Then $\D q^1 \wedge \D q^2$ is simple while $\D q^1 \wedge \D q^2 + \D q^3 \wedge \D q^4$ is not simple. 

\begin{definition}\label{simple representation of M-A equation}
Consider a M-A equation $\Delta_\omega \phi = 0$. Then $\omega$ is called a simple representation of the equation, if it has constant coefficients and contains the minimal number of simple terms. 
\end{definition}

\begin{remark}
Note that the property of being simple is basis dependent. On the other hand, the effectivity is a basis independent notion.
\end{remark}

It seems natural to denote Lagrangian functions and their corresponding counterpart defined on $J^1M$ by the same symbol, i.e. to write $L = L(q^\mu, u, p_\mu)$ as well as $(j^1 \phi)^* L = L (q^\mu, \phi, \phi_\mu)$. To avoid any confusion, we distinguish the two in the following proposition as follows. A Lagrangian function that can be integrated over $M$ will be $L$, its $J^1M$ counterpart will be $\tilde{L}$.

%%%%%%%%%%%%%%%%%%%%%%%%%%%%%%%%%%%%%%%%%%%%%
\begin{proposition}\label{necessary condition for the existence of a solution to the inverse variational problem}
Let $\Delta_\omega \phi = 0$ be a M-A equation over an open subset of a smooth manifold $M$, $\dim M = n$. Then a necessary condition for a first-order Lagrangian function $L = L(q^\mu, \phi, \phi_\mu)$ to be a local solution of the inverse variational problem corresponding to $\Delta_\omega \phi = 0$ is 
    \begin{equation}\label{first-order Lagrangians - necessary condition}
        k \omega_\epsilon = \mathcal{E}(\tilde{L} \beta) \ ,
    \end{equation}
for some non-vanishing function $k \colon J^1M \to \mathbb{R}$, where $\omega_\epsilon$ is the effective part of $\omega$, $\mathcal{E}$ is the Euler operator given by \eqref{Euler operator}, $\tilde{L} \colon J^1M \to \mathbb{R}$ is such that $\tilde{L} \circ j^1 \phi = L (q^\mu, \phi, \phi_\mu) $, and $\beta = \D q^1 \wedge \ldots \D q^n$.  
\end{proposition}

\begin{proof}
Let $\alpha \in \Omega^n(J^1M)$ be a first-order Lagrangian in the sense of the definition \ref{first-order Lagrangian - definition}, i.e. $\Delta_\alpha \phi = L\beta $, for some $L$ (possibly defined only locally) which depends smoothly on $\phi$ up to the first-order in derivatives, $L = L(q^\mu, \phi, \phi_\mu) $. Assume that the E-L equation for $L$ is given by $\Delta_\omega \phi = 0$. Define 
    \begin{equation*}
        \Phi_\alpha [\phi] : = \int_M \Delta_\alpha \phi = \int_M L\beta
    \end{equation*}
(consider only $\phi$ compactly supported). Without loss of generality, we may restrict $\alpha$~to be effective (see the discussion in the subsection with effective forms) and thus by proposition \ref{first-order Lagrangians - proposition}, we (locally) have $\alpha = \tilde{L}\beta$ for appropriate $\tilde{L} \in C^\infty (J^1M)$ satisfying $\tilde{L} \circ j^1 \phi = L$. Thus $\Phi_\alpha [\phi] = \Phi_{\tilde{L}\beta} [\phi]$ and, by the second statement of lemma \ref{lemma about the Euler operator end E-L equations}, we know that the E-L equation for the functional $\Phi_{\tilde{L}\beta} [\phi]$ is $\Delta_{\mathcal{E}(\tilde{L} \beta)} \phi = 0$. Since we assumed that $L$ locally solves the inverse variational problem given by the equation $\Delta_\omega \phi = 0$, and because $\omega$ and $\omega_\epsilon$ determine the same equation, we have 
    \begin{equation*}
        \Delta_{\omega_\epsilon} \phi = 0 \iff \Delta_{\mathcal{E}(\tilde{L} \beta)} \phi = 0 \ . 
    \end{equation*}
By the first statement of lemma \ref{lemma about the Euler operator end E-L equations}, $\mathcal{E}(\tilde{L} \beta)$ is an effective form. Since $\omega_\epsilon$ and $\mathcal{E}(\tilde{L} \beta)$ are effective forms determining the same equation, the corollary \ref{lemma for the necessary condition theorem} implies that the forms must differ by a multiple of a~non-vanishing function.
\end{proof}

\begin{remark}
Although we work locally in a coordinate system, notice that the necessary conditions for the existence of a solution to the inverse variational problem is, in our framework, a tensorial statement and thus independent of the choice of coordinates. 
\end{remark}

We present the following, simple example in $\dim M = 2$ to show how the proposition \ref{necessary condition for the existence of a solution to the inverse variational problem} can be used. 

\begin{example}{Example} \label{example}
Consider the $1$D wave equation (understand one of the two coordinates as time)
    \begin{equation}\label{wave equation in dim 1}
        \phi_{11} - c\phi_{22} = 0 \ ,
    \end{equation}
where $c > 0$ is a real constant, $\phi \colon M \to \mathbb{R}$, and $\dim M = 2$. We want to find $L(q^\mu, u, p_\mu) \in C^\infty (J^1U)$ s.t. the E-L equation for $(j^1\phi)^* L = L(q^\mu, \phi, \phi_\mu)$ is \eqref{wave equation in dim 1}.

The simple representation is 
    \begin{equation*}
        \omega = - c \D q^1 \wedge \D p_2 - \D q^2 \wedge \D p_1 \ . 
    \end{equation*}
We can easily see that $\Delta_\omega \phi = 0$ gives the original equation
    \begin{equation*}
        (j^1\phi)^* \omega = - c \D q^1 \wedge \D \phi_2 - \D q^2 \wedge d\phi_1 = (\phi_{11} - c \phi_{22}) \D q^1 \wedge \D q^2 \ . 
    \end{equation*}
The simple representation is effective, $\omega = \omega_\epsilon$, since it degenerates along $\chi$
    \begin{equation*}
        \chi \mathbin{\lrcorner} \omega = \partial_u \mathbin{\lrcorner} (- c \D q^1 \wedge \D p_2 - \D q^2 \wedge \D p_1) = 0 \ ,
    \end{equation*}
and belongs to the kernel of the bottom operator 
    \begin{equation*}
        \bot \omega = \partial_{p_\mu} \mathbin{\lrcorner} \partial_{q^\mu} \mathbin{\lrcorner} (- c \D q^1 \wedge \D p_2 - \D q^2 \wedge \D p_1) = c \partial_{p_1} \mathbin{\lrcorner} ( -\D p_2 ) - \partial_{p_2} \mathbin{\lrcorner} \D p_1 = 0 \ .
    \end{equation*}
The coordinate expression of the Euler operator evaluated on a general first-order Lagrangian $n$-form is given by \eqref{E (L beta) in local coordinates}. For $n = 2$ we have $\beta = \D q^1 \wedge \D q^2$ and $\beta_1 = \partial_{q^1} \mathbin{\lrcorner} \beta = \D q^2, \beta_2 = \partial_{q^2} \mathbin{\lrcorner} \beta = -\D q^1$, so \eqref{E (L beta) in local coordinates} becomes
    \begin{eqnarray*}
        \mathcal{E} (L \beta) & = \frac{\partial^2 L}{\partial {p_1}^2} \D q^2 \wedge d p_1 + \frac{\partial^2 L}{\partial p_1 \partial p_2} \D q^2 \wedge d p_2 - \frac{\partial^2 L}{\partial p_2 \partial p_1} \D q^1 \wedge d p_1 - \frac{\partial^2 L}{\partial {p_2}^2} \D q^1 \wedge d p_2 \\
        & \ \   - ( \frac{\partial^2 L}{\partial q^1 \partial p_1 } + \frac{\partial^2 L}{\partial q^2 \partial p_2 } +  p_1 \frac{\partial^2 L}{\partial u \partial p_1} + p_2 \frac{\partial^2 L}{\partial u \partial p_2} + \frac{\partial L}{\partial u} ) \D q^1 \wedge \D q^2 
    \end{eqnarray*}
We can fix the value of the function in \eqref{first-order Lagrangians - necessary condition} to be constant, say $k = 1$, since two forms which are multiple of each other by a smooth non-vanishing $k$ yields the same M-A equation. Hence we search for $L \in C^\infty(J^1M)$ such that $\omega = \mathcal{E} (L \beta)$, which implies 
    \begin{eqnarray*}
        \frac{\partial^2 L}{\partial {p_1}^2} = - 1 \ , \ \  &  \frac{\partial^2 L}{\partial {p_2}^2} = c \ , \ \ & \frac{\partial L}{\partial u} = \frac{\partial L}{\partial q^\mu} = 0 \ , \   \mu = 1,2 \ . 
    \end{eqnarray*}
Thus $L = L(p_\mu)$ and we can solve the first two conditions by the choice 
    \begin{equation*}
        L = \frac{1}{2}(-{p_1}^2 + c {p_2}^2) \ . 
    \end{equation*}
because the M-A equation $\Delta_{\mathcal{E} (L \beta)} \phi = 0$ writes
    \begin{equation*}
        \frac{\partial (j^1\phi)^* L}{\partial \phi} -  \frac{\partial}{\partial q^\mu} \frac{\partial (j^1\phi)^* L}{\partial \phi_\mu} = \phi_{11} - c \phi_{22} = 0 \ .
    \end{equation*}
We see that $(j^1\phi)^* L = \frac{1}{2}(-{\phi_1}^2 + c {\phi_2}^2)$ is a solution to the inverse problem for \eqref{wave equation in dim 1}. \qed    
\end{example}

%%%%%%%%%%%%%%%%%%%%%%%%%%%%%%%%%%%%%%%%%%%%%
\subsection{Plebański, Grant, and Husain equations.}

Proceeding in a similar fashion as in the previous example, we analysed both Plebański heavenly, Grant, and Husain equations in $\dim = 4$. The following tables summarize simple representations, show their non-effectivity and display effective parts of the simple representations of the aforementioned PDEs, $\phi$ being a real function. Since the effective forms of M-A equations in four dimensions tend to have lengthy expressions, we introduce the following shorthand notation, which also facilitate the computations. We denote
    \begin{align*}
        d^\mu & : = \D q^\mu, &  d_\mu & := \D p_\mu \ ,    
    \end{align*}
and for the wedge product, we write 
    \begin{align*}
        d^\mu_{\ \nu } & : = \D q^\mu \wedge \D p_\nu, & d^{\ \mu}_\nu & : = \D p_\nu \wedge \D q^\mu \ .
    \end{align*}
Notice that the position and order of indices matter and there are obvious relations such as $ d^\mu_{\ \nu } = - d^{\ \mu}_\nu$, or for the contractions $\partial_{q^\mu} \mathbin{\lrcorner} d^\nu = \delta^\mu_\nu$ (the Kronecker delta) and $\partial_{q^\mu} \mathbin{\lrcorner} d_\nu = \partial_{p_\mu} \mathbin{\lrcorner} d^\nu = 0$, et cetera. For example, the symplectic form is in the above notation written as $\Omega = d^1_{\ 1} + \ldots + d^n_{\ n}$, the volume form on $M$ is $\beta = d^{1234}$, and so on.

\begin{table}\label{table 1}
\begin{center}
\caption{Simple representations (which are not effective, $\bot \omega \neq 0$) of 1st Plebański (P1), 2nd Plebański (P2), Grant (G), and Husain (H) equations.} 
\begin{tabular}{ | p{2.55cm} | p{4.7cm} | p{4cm} | }
\hline & Monge-Ampère equation & simple representation  \vspace{3pt} \\ \hline 
 1st Plebański & $\phi_{13} \phi_{24}  - \phi_{14}\phi_{23} = 1$ & $\omega_{P1} = d^{12}_{\ \ 12} - d^{12}_{\ \ 34}$ \vspace{3pt} \\ 
 2nd Plebański & $\phi_{11} \phi_{22}  - (\phi_{12})^2 + \phi_{13} +  \phi_{24} = 0$ &  $\omega_{P2} = d^{123}_{\ \ \ 2} - d^{124}_{\ \ \ 1} + d^{34}_{\ \ 12}$  \vspace{3pt} \\
 Grant & $\phi_{11} + \phi_{24} \phi_{13} - \phi_{23} \phi_{14} = 0$ & $\omega_G = - d^{234}_{\ \ \ 1} - d^{12}_{\ \ 12} $  \vspace{3pt}\\
 Husain & $ \phi_{13} \phi_{24} - \phi_{14} \phi_{23} + \phi_{11} + \phi_{22} = 0$ & $ \omega_H = d^{134}_{\ \ \ 2} - d^{234}_{\ \ \ 1} + d^{12}_{\ \ 12}$  \vspace{3pt}\\ \hline
\end{tabular}
\end{center}
\end{table}

\begin{table}\label{table 2}
\begin{center}
\caption{Effective parts of simple representations of P1, P2, G, and H.} 
\begin{tabular}{ | p{2.55cm} | p{8.7cm} |  }
\hline & effective form $\omega_\epsilon$ \vspace{3pt} \\ \hline 
1st Plebański & $\omega_{P1\epsilon} = - d^{1234} + \frac{1}{3}( d^{12}_{\ \ 12} + d^{34}_{\ \ 34}) - \frac{1}{6}( d^{13}_{\ \ 13} + d^{14}_{\ \ 14} + d^{23}_{\ \ 23} + d^{24}_{\ \ 24} )$ \vspace{3pt} \\ 
2nd Plebański & $\omega_{P2\epsilon} = \frac{1}{2} (d^{124}_{\ \ \ 1} + d^{123}_{\ \ \ 2} + d^{234}_{\ \ \ 3} + d^{134}_{\ \ \ 4}) + d^{34}_{\ \ 12} $  \vspace{3pt} \\
Grant & $\omega_{G\epsilon} = - d^{234}_{\ \ \ 1} + \frac{1}{3}( d^{12}_{\ \ 12} + d^{34}_{\ \ 34} )  - \frac{1}{6}( d^{13}_{\ \ 13} + d^{14}_{\ \ 14} + d^{23}_{\ \ 23} + d^{24}_{\ \ 24} )$ \vspace{3pt}\\
Husain & $\omega_{H\epsilon} = d^{134}_{\ \ \ 2} - d^{234}_{\ \ \ 1} +  d^{12}_{\ \ 12} + d^{34}_{\ \ 34}  - \frac{1}{2}( d^{13}_{\ \ 13} + d^{14}_{\ \ 14} + d^{23}_{\ \ 23} + d^{24}_{\ \ 24})$  \vspace{3pt}\\ \hline
\end{tabular}
\end{center}
\end{table}

%%%%%%%%%%%%%%%%%%%%%%%%%%%%%%%%%%%%%%%%%%%%%

Proposition \ref{necessary condition for the existence of a solution to the inverse variational problem} yields the following result. 
\begin{corollary}
Monge-Ampère equations from table \ref{table 1} do not correspond to a variational problem of a first-order Lagrangian function. 
\end{corollary}

\begin{proof}
Table \ref{table 2} shows the effective forms of Monge-Ampère equations under consideration. In all cases, the effective form contains at least one term of the form $d^{\mu \nu}_{\ \ \xi \eta}$. These terms do not occur in the expression \eqref{E (L beta) in local coordinates}. Thus the necessary condition for the existence of a~first-order Lagrangian, given by the proposition \ref{necessary condition for the existence of a solution to the inverse variational problem}, is not satisfied. 
\end{proof}

We want to emphasize here that although the Plebański heavenly, Grant, and Husain equations do not have a first-order Lagrangian for which they would be E-L equations, in a different setup a Lagrangian can be found \cite{Sheftel-Multi-hamiltonian-Plebanski2nd, Nutku_1996}. Let us consider the second heavenly equation
    \begin{equation}\label{second heavenly}
        \phi_{11} \phi_{22}  - (\phi_{12})^2 + \phi_{13} +  \phi_{24} = 0 \ .
    \end{equation}
If we single-out one coordinate among $q^1, \dots, q^4$, say $q^1$, and introduce a new function $\psi$, then we can write \eqref{second heavenly} as an evolution system in $q^1$
    \begin{align}\label{reduced Plebanski2}
        \psi - \phi_1 & = 0 \ , \\
        \psi_1 \phi_{22} - {\psi_2}^2 + \psi_3 + \phi_{24} & = 0 \ ,
    \end{align}
Interestingly, the above system is a variational problem, since it is given by the E-L equations 
    \begin{eqnarray*}
        \frac{\partial L}{\partial \phi} - \frac{\partial }{\partial q^{\mu}} \frac{\partial L}{\partial \phi_{\mu}} + \frac{\partial^2 }{\partial q^{\mu} \partial q^{\nu}} \frac{\partial L}{\partial \phi_{\mu \nu} } & = 0 \ , \\ 
        \frac{\partial L}{\partial \psi} - \frac{\partial }{\partial q^{\mu}} \frac{\partial L}{\partial \psi_{\mu}} + \frac{\partial^2 }{\partial q^{\mu} \partial q^{\nu}} \frac{\partial L}{\partial \psi_{\mu \nu} } & = 0 \ , 
    \end{eqnarray*}
of the functional 
    \begin{equation}
    L [\phi, \psi] = \psi \phi_1 \phi_{22} + \frac{1}{2} \phi_1 \phi_3 - \frac{1}{2} \psi^2 \phi_{22} + \frac{1}{2} \phi_2 \phi_4 \ .
    \end{equation}
In \cite{Nutku_1996}, a method for treating the general case of Monge-Ampère equations is provided, together with systematic approach of finding Lagrangians for them after the decomposition into an evolution system. For further details regarding the above case, see \cite{Sheftel-Multi-hamiltonian-Plebanski2nd}. 

The following example shows an equation which has a first-order Lagrangian, the corresponding effective form does not have constant coefficients, and is not a differential form over the cotangent bundle. We will see that the conditions of proposition \ref{necessary condition for the existence of a solution to the inverse variational problem} are satisfied.

%%%%%%%%%%%%%%%%%%%%%%%%%%%%%%%%%%%%%%%%%
\subsection{Klein-Gordon equation.}

Let $M$ be a four-dimensional Minkowski spacetime with coordinates $q^\mu$ and flat metric $\eta_{\mu \nu}$ with signature $(+,-,-,-)$. Consider the (linear) Klein-Gordon equation
    \begin{equation}\label{Klein-Gordon equation}
        \phi_{11} - \phi_{22} - \phi_{33} - \phi_{44} + m^2 \phi^2 = 0 \ , 
    \end{equation} 
where $m \in \mathbb{R}$ is a constant. We can describe \eqref{Klein-Gordon equation} as a M-A equation $\Delta_\omega \phi = 0$ via the form
    \begin{equation*}
        \omega = -\beta_1 \wedge \D p_1 + \sum_{\mu = 2}^4 \beta_\mu \wedge \D p_\mu + m^2 u \beta \ . 
    \end{equation*}
This $4$-form is not a simple representation of \eqref{Klein-Gordon equation}, due to the non-constant coefficient $ m^2 u$, but it is an effective form, see the example \ref{effective forms - example}. Comparing $\omega$ with the local form of $\mathcal{E}(L\beta)$ for general $L$ (see \eqref{E (L beta) in local coordinates}), we obtain the following set of conditions
    \begin{eqnarray*}
        \eta_{\mu \nu} & = \frac{\partial^2 L}{\partial p_\mu \partial p_\nu}  ,\  \mu, \nu = 1, \ldots 4  \ ,  \\ 
        - m^2 u & = \frac{\partial^2 L}{\partial q^\mu \partial p_\mu} + p_\mu \frac{\partial^2 L}{\partial u \partial p_\mu} - \frac{\partial L}{\partial u} \ .
    \end{eqnarray*}
One can easily check that the function $L$
    \begin{equation*}
        L = \frac{1}{2}(-{p_1}^2 + \sum_{\mu = 2}^4 {p_\mu}^2 +  m^2 u^2)  \in C^\infty(J^1M)
    \end{equation*}
satisfies all the above conditions. It follows that 
    \begin{equation*}
         (j^1 \phi)^* L = \frac{1}{2}(-{\phi_1}^2 + \sum_{\mu = 2}^4 {\phi_\mu}^2 +  m^2 \phi^2)
    \end{equation*}
is a first-order Lagrangian for the Klein-Gordon equation. 

%%%%%%%%%%%%%%%%%%%%%%%%%%%%%%%%%%%%%%%%%%%%%

%%%%%%%%%%%%%%%%%%%%%%%%%%%%%%%%%%%%%%%%%%%%%%%%%%%%%%%%%%%%%%%%%%%%%%%%%%%%%%%%%%%%%%%%
%%%%%%%%%%%%%%%%%%%%%%%%%%%%%%%%%%%%%%%%%%%%%%%%%%%%%%%%%%%%%%%%%%%%%%%%%%%%%%%%%%%%%%%%

%%%%%%%%%%%%%%%%%%%%%%%%%%%%%%%%%%%%%%%%%%%%%%%%%%%%%%%%%%%%%%%%
\section{Multisymplectic formulation}
%%%%%%%%%%%%%%%%%%%%%%%%%%%%%%%%%%%%%%%%%%%%%%%%%%%%%%%%%%%%%%%%

In \cite{multisymplectic-formalims-and-covariant-phase-space} F. Hélein provided a multisymplectic formulation of the Klein-Gordon equation \eqref{Klein-Gordon equation} (in dimension $n$) over $\mathcal{M}:= \Lambda^n T^* (M \times \mathbb{R})$, equipped with the multisymplectic form
    \begin{equation}\label{multisymplectic form of Frederic}
        \mathfrak{m} : = \D e \wedge \beta + \D p_\mu \wedge \D \phi \wedge \beta_\mu \ ,
    \end{equation}
where $e$ is a fiber coordinate of the trivial line bundle $M \times \mathbb{R} \to M$, $p_\mu$ are the cotangent coordinates, $\beta = \D q^1 \wedge \ldots \wedge \D q^n $ and $\beta_\mu = \partial_{q^\mu} \mathbin{\lrcorner} \beta$, with $q^\mu$ coordinates on a $n$-dimensional Minkowski spacetime $M$. Using \eqref{multisymplectic form of Frederic}, the following Hamiltonian function on $\mathcal{M}$ is defined in such a way to correspond to solutions of \eqref{Klein-Gordon equation}
    \begin{equation*}
        \mathcal{H}:= e + \frac{1}{2} \eta_{\mu \nu} p_\mu p_\nu + \frac{1}{2} m^2 \phi^2 \ ,
    \end{equation*}
where $\eta_{\mu \nu} $ is the Minkowski metric with signature $(+,-, \ldots, -)$. 
Each solution of \eqref{Klein-Gordon equation} is then interpreted as a Hamiltonian $n$-curve, defined by equations 
    \begin{align*}
        p_\mu &  = \eta^{\mu \nu} \phi_{\nu}, \mu = 1, \ldots, n \ , \\
        e & = - \frac{1}{2} \eta^{\mu \nu} \phi_\mu \phi_\nu -  \frac{1}{2} m^2 \phi^2 \ , 
    \end{align*}
where $\eta^{\mu \nu}$ is the inverse to $\eta_{\mu \nu} $. In the aforementioned paper, F. Hélein provided a~canonical pre-quantization of the Klein-Gordon equation, and defined the notion of observables together with their brackets, which give rise to an infinite dimensional analogue of the Heisenberg algebra. The starting point of the method is the existence of a Lagrangian, which in the context of the Klein-Gordon equation is a first-order one. For more details see \cite{multisymplectic1,ObservableFormsAndFunctionals,  multisymplectic-formalims-and-covariant-phase-space}. 

%%%%%%%%%%%%%%%%%%%%%%%%%%%%%%

The following theorem is due to D. Harrivel. It enables us to associate to certain effective forms on $J^1M$ their (non-unique) multisymplectic counterpart on the trivial line bundle over $J^1M$. The proof can be found in \cite{Dika}. Note that the key difference with respect to the previous multisymplectic formulation of F. Hélein is that them multisymplectic form can be associated with Monge-Ampère equations which are not variational, that is, equations which are not Euler-Lagrange for some first-order Lagrangian. As we have seen in the previous section, this is the case for all the equations in table \ref{table 1}. 
\begin{theorem}\label{multisymplectic from effective theorem}
Let $\omega \in \Omega^n (\mathcal{C})$ be an effective form, $n = \dim M$. Consider a trivial line bundle $\mathcal{T} : = J^1M \times \mathbb{R} \to J^1M$ with fiber coordinate $e$. Define $\mathfrak{m}_\omega \in \Omega^{n+1} (\mathcal{T})$ by 
    \begin{equation}\label{multisymplectic form from the effective one}
        \mathfrak{m}_\omega : = \D e \wedge \beta + \mathfrak{c} \wedge \omega \ .
    \end{equation}
Then $\mathfrak{m}_\omega$ is a multisymplectic form if and only if
    \begin{enumerate}
        \item The set $\mathcal{S}_\omega : = \{\partial_{q^1} \mathbin{\lrcorner} \omega, \ldots, \partial_{q^n} \mathbin{\lrcorner} \omega \} $ is linearly independent over $\Omega^{n-1} (\mathcal{C})$, and,
        \item $\D_p \omega = 0$. 
    \end{enumerate}
\end{theorem}

Once an equations has a simple representation, the corresponding effective form has constant coefficients, and thus the second assumption of theorem \ref{multisymplectic from effective theorem} is trivially satisfied since $\D_p = p \circ \D$. The linear independence of the set $\mathcal{S}_\omega$ in the case of $4$D equations is decided over $\binom{\dim \mathcal{C}}{\dim M -1} = \binom{8}{3} = 56$-dimensional space of $3$-forms on $\mathcal{C}$. In all our cases, this can be determined almost without computation.

%%%%%%%%%%%%%%%%%%%%%%%%%%%%%%%%%%%%%%%%%%%%%%%%%%%%%
\subsection{Plebański, Grant, and Husain equations.} 

For the first heavenly equation we have 
    \begin{equation*}
        \mathcal{S}_{P1} = \{ -d^{234} + x , d^{134} + y, - d^{124} + z, d^{123} + w\} \ , 
    \end{equation*}
where $x,y,z,w$ are linear combinations of $d^{\mu}_{\ \nu \xi}$, for appropriate $\mu,\nu,\xi$. We see that $\mathcal{S}_{P1} $ is linearly independent. Similarly for the second heavenly equation
    \begin{eqnarray*}
        \mathcal{S}_{P2} = \{ &  \frac{1}{2}(d^{24}_{\ \ 1} + d^{23}_{\ \ 2} + d^{34}_{\ \ 4}) + d^{34}_{\ \ 2} ,  \frac{1}{2}(-d^{14}_{\ \ 1} - d^{13}_{\ \ 2} + d^{34}_{\ \ 3}) - d^{34}_{\ \ 1}, \\  & \frac{1}{2}(d^{12}_{\ \ 2} - d^{24}_{\ \ 3} - d^{14}_{\ \ 4}) + d^{4}_{\ 12},  \frac{1}{2}(d^{12}_{\ \ 1} + d^{23}_{\ \ 3} + d^{13}_{\ \ 4}) + d^{3}_{\ 12} \} \ , 
    \end{eqnarray*}
which is a linearly independent set as the simple terms are all different. It is not difficult to check that the sets $\mathcal{S}_G$ and $\mathcal{S}_H$ for Grant and Husain equations, respectively, are also linearly independent. Thus the 5-form $\mathfrak{m}_\omega$ is a multisymplectic form on $J^1M \times \mathbb{R}$ in all the four cases described in table \ref{table 1}. 

%%%%%%%%%%%%%%%%%%%%%%%%%%%%%%%%%%%%%%%%%%%%%%%%%%%%%
\subsection{Klein-Gordon equation.} 

Interestingly, and in contrast with the Plebański, Grant, and Husain equations, the $5$-form for the Klein-Gordon equation defined by \eqref{multisymplectic form from the effective one} is not a~multisymplectic form. To see this, take the differential $4$-form
    \begin{equation*}
        \omega = -\beta_1 \wedge \D p_1 + \sum_{\mu = 2}^4 \beta_\mu \wedge \D p_\mu + m^2 u \beta \ , 
    \end{equation*}
which, as we already discussed, is effective and represents \eqref{Klein-Gordon equation} as a Monge-Ampère equation $\Delta_\omega \phi = 0$. Due to the non-constant $m^2 u $ term, the exterior derivative gives 
    \begin{equation*}
        \D \omega = m^2 \D u \wedge \beta \ , 
    \end{equation*}
which is not degenerate along the Reeb field. Thus $\D_p \neq \D$ and we have to project the form down to $\Omega (\mathcal{C})$ (see \eqref{projection operator} for the definition of $p$)
    \begin{equation*}
        \D_p \omega = m^2 (\D u \wedge \beta - \mathfrak{c} \wedge \chi \mathbin{\lrcorner} ( \D u \wedge \beta) ) = m^2 p_\mu \D q^\mu \wedge \beta \ . 
    \end{equation*}
We see that the second condition of the theorem \ref{multisymplectic from effective theorem} is not satisfied and thus $\mathfrak{m}_\omega$ given by \eqref{multisymplectic form from the effective one} is not a multisymplectic form. Notice that the first condition of the theorem is not violated as the set $\mathcal{S}_\omega$ is linearly independent.

%%%%%%%%%%%%%%%%%%%%%%%%%%%%%%%%%%%%%%%%%%%%%%%%%%%%%%%%%%%%%%%%%%%%%%%%%%%%%%%%%%%%%%%%
%%%%%%%%%%%%%%%%%%%%%%%%%%%%%%%%%%%%%%%%%%%%%%%%%%%%%%%%%%%%%%%%%%%%%%%%%%%%%%%%%%%%%%%%

 %%%%%%%%%%%%%%%%%%%%%%%%%%%%%%%%%%%%
\section{Conclusion and discussion}
%%%%%%%%%%%%%%%%%%%%%%%%%%%%%%%%%%%%

In this work, we mainly focused on the following two questions. Firstly, can we decide whether a first-order Lagrangian for a given Monge-Ampère equation exists? Secondly, motivated by the work of F. Hélein \cite{multisymplectic-formalims-and-covariant-phase-space} and D. Harrivel \cite{Dika}, can we associate a multisymplectic form to equations which are not variational with respect to a first-order Lagrangian?

%%%%%%%%%%%%%%%%%%%%%%%%%%%
\hfill 

Regarding the first question, we provided a partial answer by formulating a necessary condition for the existence of a local solution to this inverse variational problem. This was done by representing a given equation by an effective differential form over the first jet space, and comparing it with an $n$-form that produces Euler-Lagrange equation for a~general, first-order Lagrangian function. 

Comparing the effective forms yields a computationally straightforward and simple method for obtaining a non-trivial information about Monge-Ampère equations in the context of strong inverse variational problems. Using the method, we showed that Plebański heavenly equations, Grant equation and Husain equation are not variational in our sense. Recall that the first heavenly equation is equivalent with the Grant equation after appropriate change of coordinates \cite{Grant1993}. Using a similar approach, we have shown (as expected) that the Klein-Gordon equation is variational by finding the well-known Lagrangian for it. The hypothesis is that the self-duality conditions imposed to derive the previous four equations creates an obstruction for the existence of the first-order Lagrangian. We want to study this problematics in more detail in our future work.

The presented method is much more suitable for deciding the non-variational nature of a given equation than solving the local inverse problem explicitly. Moreover, it works only when restricted to the case of first-order Lagrangians. Nevertheless, this limitation can be seen as desirable, since the first-order Lagrangians are of great importance throughout the physics. 

It is not clear at the moment how to generalize our approach to the case of more functions. The procedure can be naively extended for more scalar fields by introducing multiple Euler operators, the cost being degeneracy issues. This causes further problems, for example in the context of the unique decomposition of differential forms into the effective and non-effective part, which is an essential tool in our approach. In \cite{BANOS2011}, B. Banos used the notion of bi-effective forms to efficiently deal with the complex Monge-Ampère equations, and proved the possibility to always obtain a unique bieffective decomposition. This is not equivalent in an obvious way to the aforementioned naive extension, as the Verbitsky-Bonan relations are not satisfied in our case (see \cite{BANOS2011}, Theorem 1). This is connected with the fact that we do not restrict our forms to have coefficients independent of the $u$ coordinate on $J^1M$ (which allows us to work, for example, with the Klein-Gordon equation). Whether this problems can be resolved will be part of our future investigations.

%%%%%%%%%%%%%%%%%%%%%%%%%%%%%
\hfill

Regarding the second question focused on the multisymplectic formulation of Monge-Ampère equations. Using the results of \cite{Dika}, we provided multisyplectic $5$-forms in the case of real $4$-dimensional heavenly Plebański, Grant, and Husain equations, all of which are not variational in our sense. Interestingly, the same approach does not work for the Klein-Gordon equation as the corresponding $5$-form is not multisymplectic. 

F. Hélein's multisymplectic treatment of the Klein-Gordon equation provided in \cite{multisymplectic-formalims-and-covariant-phase-space} starts with a first-order Lagrangian function. The other four Monge-Ampère equations we studied cannot be treated in the same way, unless going into higher order Lagrangians. On the other hand, the theorem \ref{multisymplectic from effective theorem} provides a multisymplectic forms exactly for the four non-variational cases and fails for the Klein-Gordon equation. To provide some explanation of this, it would be interesting to compare the methods of \cite{multisymplectic1, multisymplectic-formalims-and-covariant-phase-space} with those in \cite{Dika} in the situation of a general Monge-Ampère equation.

%%%%%%%%%%%%%%%%%%%%%%%%%%%%%%%%%%%%%%%%%%%%%%%%%%%%%%%%%%%%%%%%%%%%%%%%%%%%%%%%%%%%%%%%
%%%%%%%%%%%%%%%%%%%%%%%%%%%%%%%%%%%%%%%%%%%%%%%%%%%%%%%%%%%%%%%%%%%%%%%%%%%%%%%%%%%%%%%%

%%%%%%%%%%%%%%%%%%%%%%%%%%%%%%%%%
\section*{Acknowledgment}
%%%%%%%%%%%%%%%%%%%%%%%%%%%%%%%%%

This paper was written during my visit in Angers as a part of my PhD research, under the cotutelle agreement between the Masaryk University, Brno, Czech Republic, and the University of Angers, France.  I am grateful for the funding provided by the Czech Ministry of Education, and by the Czech Science Foundation under the project GAČR EXPRO GX19-28628X, and I thank the University of Angers for the hospitality during the research period. I also want to express my gratitude to Volodya Rubtsov for his numerous valuable suggestions and detailed comments, and to Jan Slovák for clarification of concepts from the theory of jet bundles. The results were reported at the Winter School and Workshop Wisla 20-21, a European Mathematical Society event organized by the Baltic Institute of Mathematics.

%%%%%%%%%%%%%%%%%%%%%%%%%%%%%%%%%%%%%%%%%%%%%%%%%%%%%%%%%%%%%%%%%%%%%%%%%%%%%%%%%%%%%%%%
%%%%%%%%%%%%%%%%%%%%%%%%%%%%%%%%%%%%%%%%%%%%%%%%%%%%%%%%%%%%%%%%%%%%%%%%%%%%%%%%%%%%%%%%


\begin{thebibliography}{99.}%
% Use the following syntax and markup for your references if 
% the subject of your book is from the field 
% "Mathematics, Physics, Statistics, Computer Science"
%
% Contribution 

\bibitem{multisymplectic1}Hélein, F., Hamiltonian formalisms for multidimensional calculus of variations and perturbation theory. (arXiv,2002), https://arxiv.org/abs/math-ph/0212036

\bibitem{multisymplectic-formalims-and-covariant-phase-space}Hélein, F., Multisymplectic formalism and the covariant phase space. {\em Variational Problems In Differential Geometry}. pp. 94-126 (2011)

\bibitem{ObservableFormsAndFunctionals}Hélein, F. \& Kouneiher, J., The Notion of Observable in the Covariant Hamiltonian Formalism for the Calculus of Variations with Several Variables. {\em Advances In Theoretical And Mathematical Physics}. \textbf{8}, 735 - 777 (2004), https://doi.org/

\bibitem{MultisymplecticFormalismForEH}Gaset, J. \& Román-Roy, N., Multisymplectic unified formalism for Einstein-Hilbert gravity. (arXiv), https://arxiv.org/abs/1705.00569v5

\bibitem{PropertiesOfMultisympl}Román-Roy, N., Some Properties of Multisymplectic Manifolds. (arXiv), arXiv:1807.11774v2

\bibitem{firstOrderTheoriesAndPremultisymplectic}Cédric M. Campos, Elisa Guzmán, Juan Carlos Marrero, Classical field theories of first order and Lagrangian submanifolds of premultisymplectic manifolds. {\em Journal Of Geometric Mechanics}. \textbf{4} (2012)

\bibitem{Sheftel-Multi-hamiltonian-Plebanski2nd}Neyzi, F., Nutku, Y. \& Sheftel, M.: Multi-Hamiltonian structure of Plebanski's second heavenly equation. {\em J. Phys. A}. \textbf{38} pp. 8473 (2005)

\bibitem{Lychagin}Lychagin, V., Contact Geometry and Non-Linear Second-Order Differential Equations. {\em Russian Mathematical Surveys}. \textbf{34}, 149-180 (1979,2)

\bibitem{Plebanski:1975wn}Plebanski, J., Some solutions of complex Einstein equations. {\em J. Math. Phys.}. \textbf{16} pp. 2395-2402 (1975)

\bibitem{kushner_lychagin_rubtsov_2006}Kushner, A., Lychagin, V. \& Rubtsov, V., Contact Geometry and Nonlinear Differential Equations. (Cambridge University Press,2006)

\bibitem{Natural-operations}Kolář, I., P. Michor \& Slovák, J., Natural operations in differential geometry. (Springer-Verlag: Berlin Heidelberg,1993), https://www.emis.de/monographs/KSM/

\bibitem{Dika}Harrivel, D. Hamiltonian, Multisymplectic formalism and Monge-Ampère equations. {\em Systèmes Intégrables Et Théorie Quantiques Des Champs}. pp. 331-354 (2008)

\bibitem{Nutku_1996}Nutku, Y., Hamiltonian structure of real Monge - Ampère equations. {\em Journal Of Physics A: Mathematical And General}. \textbf{29}, 3257-3280 (1996,6), https://doi.org/10.1088/0305-4470/29/12/029

\bibitem{Husain1993SelfdualGA}Husain, V., Self-dual gravity as a two-dimensional theory and conservation laws. {\em Classical And Quantum Gravity}. \textbf{11} pp. 927-937 (1993)

\bibitem{Gordon}Gordon, W., {\em Zeitschrift Für Physik}. \textbf{40} pp. 117-133 (1926)

\bibitem{Ashtekar-Jacobson-Smolin}Ashtekar, A., Jacobson, T. \& Smolin, L., A new characterization of half-flat solutions to Einstein's equation. {\em Communications In Mathematical Physics}. \textbf{115}, 631 - 648 (1988), https://doi.org/

\bibitem{Grant1993}Grant, J., On self-dual gravity. {\em Phys. Rev. D}. \textbf{48}, 2606-2612 (1993,9), https://link.aps.org/doi/10.1103/PhysRevD.48.2606

\bibitem{BANOS2011}Banos, B., Complex solutions of Monge-Ampère equations. {\em Journal Of Geometry And Physics}. \textbf{61}, 2187-2198 (2011), https://www.sciencedirect.com/science/article/pii/S0393044011001641

\bibitem{Klein-Gordon-in-QFT}Radzikowski, M., Micro-local approach to the Hadamard condition in quantum field theory on curved space-time. {\em Communications In Mathematical Physics}. \textbf{179}, 529 - 553 (1996), https://doi.org/

\end{thebibliography}
\end{document}